\documentclass{amsart}
\usepackage{amsmath}
\usepackage{amssymb}
\usepackage{amsthm}

\newtheorem{theorem}{Theorem}

\newtheorem{lemma}{Lemma}
\newtheorem{rem}{Remark}

\begin{document}
\title{A geometric approach to Wigner-type theorems}
\author{Mark Pankov, Thomas Vetterlein}

\address{Mark Pankov: Faculty of Mathematics and Computer Science, 
University of Warmia and Mazury, S{\l}oneczna 54, 10-710 Olsztyn, Poland}
\email{pankov@matman.uwm.edu.pl}

\address{Thomas Vetterlein: Department of Knowledge-Based Mathematical Systems,
Johannes Kepler University Linz,
Altenberger Stra\ss{}e 69, 4040 Linz, Austria}
\email{Thomas.Vetterlein@jku.at}

\begin{abstract}
Let $H$ be a complex Hilbert space and let ${\mathcal P}(H)$ be the associated projective space (the set of rank-one projections). Suppose that $\dim H\ge 3$. 
We prove the following Wigner-type theorem: if $H$ is finite-dimensional, then 
every orthogonality preserving transformation of ${\mathcal P}(H)$ is induced by 
a unitary or anti-unitary operator. This statement will be obtained as a consequence of the following result: every orthogonality preserving lineation of ${\mathcal P}(H)$ to itself is induced by a linear or conjugate-linear isometry ($H$ is not assumed to be finite-dimensional). As an application, we describe (not necessarily  injective) transformations of Grassmannians preserving some types of principal angles. 
\end{abstract}

\maketitle

\section{Introduction}
It follows from Gleason's theorem \cite{Gleason} that pure states of a quantum mechanical system are precisely the rank-one projections. Identifying every projection with its image, we arrive at the projective space ${\mathcal P}(H)$ formed by $1$-dimensional subspaces of a complex Hilbert space $H$. 
Our concern is the discussion of symmetries of quantum mechanical systems. The classical version of Wigner's theorem \cite{Wigner} characterizes  them as follows: every bijective transformation of the set of pure states preserving the transition probability (the trace of the composition of two projections or, equivalently, the angle between the corresponding rays) is induced by a unitary or anti-unitary operator. Transformations of this kind are known as {\it Wigner symmetries}. More generally, an arbitrary transformation of ${\mathcal P}(H)$ preserving the angle between any pair of $1$-dimensional subspaces is induced by a linear or conjugate-linear isometry (see, for example, \cite{Geher1, Molnar-book,Pankov-book}). There is no additional assumption, but the condition on the preservation of the angles immediately implies that such a transformation is injective.
Some significant generalizations of the bijective version of Wigner's theorem are obtained in \cite{Geher3,Molnar1}.

On the other hand, Uhlhorn \cite{Uhlhorn} provided a geometric approach to Wigner's theorem
based on the Fundamental Theorem of Projective Geometry:
if the dimension of $H$ is not less than $3$, then
every bijective transformation of ${\mathcal P}(H)$ preserving the orthogonality relation in both directions is a Wigner symmetry. The proof  of this statement can be sketched as follows. 
Every bijection preserving the orthogonality relation in both directions  is a collineation of the projective space (i.e. preserves the family of lines in both directions); consequently, by the Fundamental Theorem of Projective Geometry, it is induced by a semilinear automorphism of $H$; finally, every semilinear automorphism of $H$ sending orthogonal vectors to orthogonal vectors is a scalar multiple of a unitary or anti-unitary operator.  

If $H$ is infinite-dimensional, then there is an injective transformation of ${\mathcal P}(H)$ that preserves the orthogonality relation  in both directions and that is not induced by a linear or conjugate-linear isometry \cite{Semrl}. Hence the bijectivity assumption  cannot be omitted in Uhlhorn's version of Wigner's theorem.  In the present paper, we prove the following Wigner-type theorem (Theorem \ref{theorem-wig}): if the dimension of $H$ is finite and not less than $3$, then an arbitrary orthogonality preserving transformation of ${\mathcal P}(H)$ (which sends orthogonal rays to orthogonal rays without the assumption that the orthogonality relation is preserved in both directions) is a Wigner symmetry.

Our basic observation is the following. 
If $H$ is finite-dimensional, then every orthogonality preserving transformation of ${\mathcal P}(H)$ is a {\it lineation} which means that it sends every line to a subset of a line. In general, the behavior of lineations between projective spaces is complicated; they are not injective and can send lines to parts of lines only. Our version of Wigner's theorem is a consequence of the following result (Theorem \ref{theorem-lin}): every orthogonality preserving lineation of ${\mathcal P}(H)$ to itself is induced by a linear or conjugate-linear isometry; as above, we assume that the dimension of $H$ is not less than $3$, but do not require that $H$ is finite-dimensional.  
We note that orthogonality preserving lineations between the projective spaces associated to anisotropic Hermitian spaces are investigated in \cite{PV}.

The proof of Theorem \ref{theorem-lin} will be given in two steps. Using Gleason's theorem, we establish that every orthogonality preserving lineation is {\it non-degenerate};  in particular, the image of every line contains at lest three elements. This guarantees that such a lineation is induced by a generalized semilinear map (a modification of the Fundamental Theorem of Projective Geometry \cite{Faure,Mach,Rado}). Our next step is to show that orthogonality preserving generalized semilinear maps are precisely linear and conjugate-linear isometries, which is equivalent to the fact that every place of the complex field ${\mathbb C}$ (a homomorphism of a valuation ring of ${\mathbb C}$ to ${\mathbb C}$) is the identity or the complex conjugation. 

The conjugacy class of rank-$k$ projections can be naturally identified with the Grassmannian consisting of $k$-dimensional subspaces of $H$. Wigner's theorem was extended on such Grassmannias by Moln\'ar \cite{Molnar2, Molnar3} (see \cite{Geher2,Gyory,Semrl} for closely connected results). The transformation of Grassmannians induced by linear and conju\-gate-linear isometries are characterized as transformations preserving principal angles between subspaces.  
By \cite{Pank-1}, it is sufficient to require that only some of principal angles (related to adjacency and orthogonality) are preserved. Using Theorem \ref{theorem-lin}, we prove a non-injective version of this result (Theorem \ref{theorem-wig-gr}).

\section{Results}
Throughout the paper we assume that
$H$ is a complex Hilbert space of dimension not less than $3$ (possibly $H$ is infinite-dimensional)
and denote by ${\mathcal P}(H)$ the associated projective space, 
i.e. the set of all $1$-dimensional subspaces of $H$.
The following statement describes orthogonality preserving transformations of ${\mathcal P}(H)$,
i.e. such that for any two orthogonal $1$-dimensional subspaces of $H$ the images are orthogonal.

\begin{theorem}\label{theorem-wig}
If $H$ is finite-dimensional, then any orthogonality preserving transformation of ${\mathcal P}(H)$
is a Wigner symmetry, i.e. a bijection induced by a unitary or anti-unitary operator on $H$.
\end{theorem}

Recall that a line in the projective space ${\mathcal P}(V)$ associated to a vector space $V$
consists of all $1$-dimensional subspaces of a certain $2$-dimensional subspace,
i.e. it is a set of type ${\mathcal P}(S)$, where $S$ is a $2$-dimensional subspace of $V$.
A {\it lineation} is a map between projective spaces which sends lines to subsets of lines.

Theorem \ref{theorem-wig} is an immediate consequence of Theorem \ref{theorem-lin},
which will be presented below, and the following lemma.

\begin{lemma}\label{lemma-1}
If $H$ is finite-dimensional, then every orthogonality preserving transformation of ${\mathcal P}(H)$
is a lineation.
\end{lemma} 

\begin{proof}
Suppose that $\dim H=n$ is finite.
Let $f$ be an orthogonality preserving transformation of ${\mathcal P}(H)$
and let $S$ be a $2$-dimensional subspace of $H$. 
We choose mutually orthogonal $1$-dimensional subspaces $P_1,\dots, P_{n-2}$
whose sum coincides with $S^{\perp}$. 
The $1$-dimensional subspaces $f(P_1),\dots, f(P_{n-2})$ are mutually orthogonal
and the orthogonal complement of their sum is a $2$-dimensional subspace $S'$. 
Every $1$-dimensional subspace $P\subset S$ is orthogonal to each $P_i$
and, consequently,  $f(P)$ is orthogonal to every $f(P_{i})$ which implies that 
$f(P)\subset S'$. 
So, $f({\mathcal P}(S))$ is contained in ${\mathcal P}(S')$.
\end{proof}

\begin{rem}{\rm
If $H$ is infinite-dimensional, then the above statement holds only in the case when
a transformation of ${\mathcal P}(H)$
sends maximal collections of mutually orthogonal $1$-dimensional subspaces 
to  maximal collections of mutually orthogonal $1$-dimensional subspaces. 
}\end{rem}

Let $V$ be a vector space over a field $F$ and $\dim V\ge 3$.
A map $L:V\to V$ is {\it semilinear} if 
$$L(x+y)=L(x)+L(y)$$
for all $x,y\in V$ and there is a non-zero homomorphism $\sigma:F\to F$ 
such that 
$$L(ax)=\sigma(a)L(x)$$
for all $a\in F$ and $x\in V$. 
Every semilinear injection $L:V\to V$ induces 
a lineation of the projective space ${\mathcal P}(V)$ to itself,
which sends a $1$-dimensional subspace $P\subset V$
to the $1$-dimensional subspace containing $L(P)$.
This lineation is not necessarily injective (see, for example, \cite[Section 2.1]{Pankov-book}).
The following version of the Fundamental Theorem of Projective Geometry is well-known: 
every injective lineation of ${\mathcal P}(V)$ to itself 
whose image is not contained in a line is induced by a semilinear injective transformation of $V$
\cite{FF, Havlicek} (see also \cite{Pankov-book}).

If an injective semilinear  transformation of $H$ is orthogonality preserving, 
then it is a scalar multiple of a linear or conjugate-linear isometry
(see, for example, \cite[Proposition 4.2]{Pankov-book}).
In the case when $H$ is finite-dimensional, such isometries are precisely unitary and anti-unitary operators.

\begin{theorem}\label{theorem-lin}
Every orthogonality preserving lineation $f:{\mathcal P}(H)\to {\mathcal P}(H)$ 
is induced by a linear or conjugate-linear isometry of $H$ to itself.
\end{theorem}

In Theorem \ref{theorem-lin}, we do not assume that $H$ is finite-dimensional. 
Theorem \ref{theorem-wig} is a direct consequence of Theorem \ref{theorem-lin} and Lemma \ref{lemma-1}.

\section{Proof of Theorem  \ref{theorem-lin}}

A lineation is said to be {\it non-degenerate} if the following conditions are satisfied:
\begin{enumerate}
\item[(L1)] the image of the lineation is not contained in a line;
\item[(L2)] the image of every line contains at least three points.
\end{enumerate}

\begin{lemma}\label{lemma-nondeg}
Every orthogonality preserving lineation $f:{\mathcal P}(H)\to {\mathcal P}(H)$  is non-degenerate.
\end{lemma}

\begin{proof}
The condition (L1) follows immediately from the fact that $f$ is orthogonality preserving.
Since each line of ${\mathcal P}(H)$ contains a pair 
of orthogonal $1$-dimensional subspaces, the image of every line contains at least two orthogonal elements.

Suppose that ${\mathcal P}(S)$ is a line whose image consists of two elements.
Let $P$ and $Q$ be orthogonal $1$-dimensional subspaces of $S$.
The image of every $1$-dimensional subspace of $S$ coincides with $P'=f(P)$ or $Q'=f(Q)$.
We take any $1$-dimensional subspace $T\subset H$ orthogonal to both $P,Q$ and 
consider the restriction of $f$ to ${\mathcal P}(M)$, where $M=P+Q+T$.
The $1$-dimensional subspaces $P',Q'$ and $T'=f(T)$ are mutually orthogonal.

Every $1$-dimensional subspace of $M$ is contained in a $2$-dimensional subspace $T+N$,
where $N$ is a $1$-dimensional subspace of $S$.
Since $f$ is a lineation and $f({\mathcal P}(S))=\{P',Q'\}$, 
the image of ${\mathcal P}(M)$ is contained in the union of 
the lines $${\mathcal P}(P'+T')\;\mbox{ and }\;{\mathcal P}(Q'+T').$$
One of the following possibilities is realized:
\begin{enumerate}
\item[(1)] the image of every $1$-dimensional subspace of $P+T$ is $P'$ or $T'$
and the image of every $1$-dimensional subspace of $Q+T$ is $Q'$ or $T'$;
\item[(2)] for one of the lines ${\mathcal P}(P+T)$ or ${\mathcal P}(Q+T)$,
say ${\mathcal P}(P+T)$, the image contains more than two elements.
\end{enumerate}

Case (1). 
For every $1$-dimensional subspace $N\subset M\setminus S$ we take
the $1$-dimensional subspaces $N_P \subset P+T$ and $N_Q \subset Q+T$ such that 
$N$ is the intersection of the $2$-dimensional subspaces $P+N_Q$ and $Q+N_P$.
By our assumption, 
$$f(N_P)\in \{P',T'\}\;\mbox{ and }\;f(N_Q)\in \{Q',T'\}.$$
Using the fact that $f$ is a lineation, we establish that 
$$f({\mathcal P}(M))=\{P',Q',T'\}.$$
This means that $f$ induces a two-valued measure on ${\mathcal P}(M)$.
Indeed, we assign $1$ to one of $P',Q',T'$ and $0$ to the remaining two
and observe that for any three mutually orthogonal $1$-dimensional subspaces of $M$
only one of these subspaces corresponds to $1$.
Gleason's theorem provides a descriptions of all measures on 
the projective spaces associated to complex Hilbert spaces, in particular,
it shows that there is no two-valued measure, i.e. the case (1) is not realized.

Case (2). By assumption, there is a $1$-dimensional subspace 
$N_0 \subset P+T$ such that $f(N_0)\ne P', T'$.
For every $1$-dimensional subspace $N\subset Q+T$ 
the lineation $f$ transfers the intersection of the lines ${\mathcal P}(N_{0}+N)$
and ${\mathcal P}(S)$ to $P'$ or $Q'$. Since $f(N_0)\ne P', T'$, this implies 
that $f(N)$ is $Q'$ or $T'$. So, the image of the line ${\mathcal P}(Q+T)$ is $\{Q',T'\}$.
We have established that 
$$f({\mathcal P}(M))\subset{\mathcal P}(P'+T')\cup \{Q'\}.$$
Now, we present the line ${\mathcal P}(P'+T')$
as the disjoint union of two subsets ${\mathcal X}_1$ and ${\mathcal X}_2$
such that any two elements from each ${\mathcal X}_{i}$, $i\in \{1,2\}$ are non-orthogonal.
Let $g$ be the transformation of ${\mathcal P}(P'+Q'+T')$
defined as follows: it sends the elements of ${\mathcal X}_1$ and ${\mathcal X}_2$
to $P'$ and $T'$ (respectively) and the remaining elements go to $Q'$.
The composition $gf|_{{\mathcal P}(M)}$ is an orthogonality preserving 
transformation whose image is formed by $P',Q',T'$.
So, the case (2) is reduced to the case (1) and, consequently,   it also is  impossible. 
\end{proof}

Now, we describe the concepts of place and generalized semilinear map.  
For a field $F$ the additive and multiplicative operations 
can be extended to a partial operation on $F\cup\{\infty\}$ as follows:
$a+\infty=\infty$ for every $a\in F$ and $a\cdot \infty=\infty$ for every non-zero $a\in F$
(note that $\infty +\infty$ and $0\cdot \infty$ are undefined).
Similarly, for a vector space $V$ over $F$ we extend the additive operation and the scalar multiplication 
in the following way: $x+\infty=\infty$ for every vector $x\in V$ and 
$a\cdot\infty=\infty\cdot x=\infty$ for every scalar $a\in (F\setminus\{0\})\cup \{\infty\}$
and every vector $x\in (V\setminus\{0\})\cup \{\infty\}$
(again, $\infty +\infty$ and $0\cdot \infty$ are undefined). 

A {\it place} of $F$ is a map 
$$\sigma: F\to F\cup\{\infty\}$$
such that $\sigma(1)=1$ and
$$\sigma(a+b)=\sigma(a)+\sigma(b),\;\;\;\sigma(ab)=\sigma(a)\sigma(b)$$ 
provided the second sum and the second product are defined.
Then $R=\sigma^{-1}(F)$ is a {\it valuation ring} of $F$, which means that
for every non-zero $a\in F$ we have $a\in R$ or $a^{-1}\in R$.
Also, the ideal 
$$I_{R}=\{a\in R: a=0\;\mbox{or}\;a^{-1}\not\in R\}$$
is the kernel of $\sigma$.

Let $V$ be a vector space over $F$ and $\dim V\ge3$.
A map 
$$L:V\to V\cup \{\infty\}$$
is called a {\it generalized semilinear map} if it satisfies the following conditions:
\begin{enumerate}
\item[$\bullet$] $L(x+y)=L(x)+L(y)$ provided the second sum is defined;
\item[$\bullet$] there is a place $\sigma$ of $F$  such that 
$L(ax)=\sigma(a)L(x)$ provided the second product is defined;
\item[$\bullet$] $L(0)=0$.
\end{enumerate}
In this case, $M=L^{-1}(V)$ is a submodule of $V$ over the valuation ring $R=\sigma^{-1}(F)$.
Suppose that the following condition is satisfied:
\begin{enumerate}
\item[(*)]for every $1$-dimensional subspace $P\subset V$ there is $x\in P\cap M$ such that $L(x)\ne 0$.
\end{enumerate}
Then $L$ induces a lineation of ${\mathcal P}(V)$ to itself
which sends every $1$-dimensional subspace $P\subset V$
to the $1$-dimensional subspace containing $L(P\cap M)$.
Every non-degenerate lineation of ${\mathcal P}(V)$ to itself
is induced by a generalized semilinear map satisfying (*), see \cite{Faure,Mach,Rado}.

\begin{lemma}\label{lemma-cp1}
Let $\sigma$ be a place of the complex field ${\mathbb C}$ 
and let $R$ be the associated valuation ring.
Suppose that there is a generalized semilinear map $L:H\to H\cup\{\infty\}$ 
over $\sigma$ which satisfies {\rm (*)} and induces an orthogonality preserving lineation.
Then the following assertions are fulfilled: 
\begin{enumerate}
\item[(1)] $I_R$ is closed under complex conjugation;
\item[(2)] $R$ is closed under complex conjugation;
\item[(3)] $\sigma(\overline{a})=\overline{\sigma(a)}$ for every $a\in R$.
\end{enumerate}
\end{lemma}

\begin{proof}
Let $M=L^{-1}(H)$. 
Since $L$ induces an orthogonality preserving lineation,
for any two orthogonal vectors from $M\setminus {\rm Ker}(L)$ the images are orthogonal.
We take any non-zero $x\in M\setminus {\rm Ker}(L)$ and any $y\in H$ orthogonal to $x$ satisfying $||y||=||x||$.
Our first claim is that $y\in M\setminus {\rm Ker}(L)$.

If $y\in {\rm Ker}(L)$, then 
$$L(x+y)=L(x-y)=L(x)\ne 0;$$
on the other hand, the vectors $x+y,x-y$ are orthogonal and we come to a contradiction. 
Assume that $y\not\in M$.  Then, by (*), we have $ay\in M\setminus {\rm Ker}(L)$ for a certain $a\in {\mathbb C}$.
In this case, $a\in I_{R}$ (otherwise, $a^{-1}\in R$ and $y=a^{-1}ay\in M$, contradicting the assumption).
Therefore, $ax,ay\in M$ and 
$$L(ay+ax)=L(ay-ax)=L(ay)\ne 0;$$
but the vectors $ay+ax,ay-ax$ are orthogonal and we get a contradiction again. 

(1). Let $a$ be a non-zero element of $I_R$. 
Assume that $\overline{a}\not\in R$.
Then $\overline{a^{-1}}=\overline{a}^{-1}$ belongs to $I_R$.
The vectors $ax-y$ and $\overline{a^{-1}}x+y$ are orthogonal and belong to $M\setminus {\rm Ker}(L)$.
Consequently, 
$$L(ax-y)=-L(y)\;\mbox{ and }\; L(\overline{a^{-1}}x+y)=L(y)$$
are orthogonal which is impossible.
Furthermore, $ax-y$ and $x+\overline{a}y$ are orthogonal vectors from $M\setminus {\rm Ker}(L)$
which implies that 
$$L(y-ax)=L(y)\;\mbox{ and }\; L(x+\overline{a}y)=L(x)+\sigma(\overline{a})L(y)$$
are orthogonal.
Since $L(x),L(y)$ are orthogonal, the latter is possible only in the case when $\sigma(\overline{a})=0$,
i.e. $\overline{a}\in I_R$.

(2). If $\overline{a}\not\in R $ for a certain non-zero $a\in {\mathbb C}$, then $(\overline{a})^{-1}\in I_{R}$
and by (1) we have $a^{-1}\in I_R$, i.e. $a\not\in R$.

(3). Let $a\in R$. 
Then $ax-y, x+\overline{a}y$ are orthogonal vectors in $M\setminus {\rm Ker}(L)$ and
the vectors
$$\sigma(a)L(x)-L(y)\;\mbox{ and }\;L(x)-\sigma(\overline{a})L(y)$$
also are orthogonal.
This implies that $\sigma(a)=\overline{\sigma(\overline{a})}$.
\end{proof}

\begin{lemma}\label{lemma-rp}
Every place of the real field ${\mathbb R}$ is the identity.
\end{lemma}

\begin{proof}
Let $\sigma: {\mathbb R}\to {\mathbb R}\cup\{\infty\}$ be a place
and let $R$ be the associated valuation ring.

First of all, we establish that $\sigma$ is order preserving, i.e. 
if $a,b\in R$, then $a\le b$ implies that $\sigma(a)\le \sigma(b)$.
It is sufficient to show that, for $a \in R$, we have that $\sigma(a)>0$ if $a >0$.
This  statement is trivial for $a\in I_R$.
Let $a\in R\setminus I_R$. Then $a^{-1}\in R\setminus I_R$.
We assert that $\sqrt{a}\in R$. 
Indeed, if $\sqrt{a}\not\in R$, then $\left (\sqrt{a}\right )^{-1}=\sqrt{a^{-1}}\in I_R$
which implies that $a^{-1}=\left (\sqrt{a^{-1}}\right )^2\in I_R$, a contradiction.
Since $\sqrt{a}\in R$, we have $\sigma(a)=\sigma\left(\sqrt{a}\right)^2\ge 0$.

The equality $\sigma(1)=1$ implies that $\sigma(n)=n$ for every $n\in {\mathbb N}$.
Let $a\in R$.
If $a>1$, then $\sigma(a)\ge 1$, in particular, $\sigma(a)\ne 0$. 
If $0<a<1$, then there is  $n\in {\mathbb N}$ such that 
$an>1$ and $\sigma(a)n=\sigma(an)\ne 0$ which means that $\sigma(a)\ne 0$.
Therefore, $\sigma(a)\ne 0$ for the case when $a\ge 0$.
If $a<0$, then $a^2 >0$ and
$\sigma(a)^2=\sigma(a^2)\ne 0$ which implies that $\sigma(a)\ne 0$.

So, $\sigma(a)\ne 0$ for all non-zero $a\in R$ and $I_{R}=0$,
i.e. $R$ coincides with ${\mathbb R}$. 
It is well-known that every non-zero endomorphism of ${\mathbb R}$ is the identity. 
\end{proof}

\begin{lemma}\label{lemma-cp2}
Let $\sigma$ be a place of the complex field ${\mathbb C}$ 
and let $R$ be the associated valuation ring.
If $\sigma(\overline{a})=\overline{\sigma(a)}$ 
for all $a\in R$, then $R$ coincides with ${\mathbb C}$ and $\sigma$ is the identity or
the complex conjugation.
\end{lemma}

\begin{proof}
The intersection $R\cap {\mathbb R}$ is a valuation ring of ${\mathbb R}$
and $\sigma(R\cap {\mathbb R})\subset {\mathbb R}$.
By Lemma \ref{lemma-rp}, the ring $R\cap {\mathbb R}$ coincides with ${\mathbb R}$
and the restriction of $\sigma$ to ${\mathbb R}$ is identity.

The equality $\sigma(i)^2=\sigma(i^2)=\sigma(-1)=-1$ implies that $\sigma(i)=i$ or $\sigma(i)=-i$.
\end{proof}

Now, we are ready to prove Theorem \ref{theorem-lin}.
Let $f:{\mathcal P}(H)\to {\mathcal P}(H)$ be an orthogonality preserving lineation.
By Lemma \ref{lemma-nondeg}, $f$ is non-degenerate and, consequently, it is induced
by a generalized semilinear map $L:H\to H$.
Lemmas \ref{lemma-cp1} and \ref{lemma-cp2} guarantee that 
the associated place is the identity or the complex conjugation. 
Since for every $1$-dimensional subspace $P\subset H$
there is a vector $x\in P\cap L^{-1}(H)$ such that $L(x)\ne 0$,
the map $L$ is a linear or conjugate-linear injection.
Also, it sends orthogonal vectors to orthogonal vectors. 
Thus, $L$ is a scalar multiple of a linear or conjugate-linear isometry and Theorem \ref{theorem-lin} is proved.

\begin{rem}{\rm
It follows from Lemma \ref{lemma-rp} that every non-degenerate lineation between real projective 
spaces is induces by a linear injection.
Since Gleason's theorem holds for the real case as well, 
orthogonality preserving lineations of real projective spaces  are non-degenerate.
Therefore, every orthogonality preserving transformation of a finite-dimensional real projective space 
is induced by an orthogonal bijection.
}\end{rem}

\section{Application of Theorem \ref{theorem-lin}}
Recall that two closed subspaces of $H$ are {\it compatible} if there is an orthogonal basis of $H$ such that 
these subspaces are spanned by subsets of this basis. 
This holds if and only if the corresponding projections commute. 

Denote by ${\mathcal G}_{k}(H)$ the Grassmannian formed by $k$-dimensional subspaces of $H$.
Note that ${\mathcal G}_{1}(H)={\mathcal P}(H)$ is the projective space associated to $H$.
If the dimension of $H$ is not less than $2k$, 
then the orthogonality relation is defined on the Grassmannian ${\mathcal G}_k(H)$.
Two $k$-dimensional subspaces of $H$ are called {\it adjacent} if their intersection is $(k-1)$-dimensional,
and these subspaces are said to be {\it ortho-adjacent} if they are adjacent and compatible. 
For $k=1$ the adjacency relation is trivial (any two distinct $1$-dimensional subspaces are adjacent).

The orthogonality, adjacency and ortho-adjacency relations can be described in terms of principal angles 
\cite[Section VII.1]{Bh}. 
It is clear that two $k$-dimensional subspaces of $H$ are orthogonal if and only if 
all principal angles between them are $\frac{\pi}{2}$.
These subspace are adjacent if and only if  precisely one of the principal angles 
is non-zero; furthermore, these subspaces are ortho-adjacent only in the case when this angle is $\frac{\pi}{2}$.
Transformations of Grassmannians preserving the principal angles are described in \cite{Molnar2, Molnar3}.
Using Theorem \ref{theorem-lin} we prove the following.

\begin{theorem}\label{theorem-wig-gr}
Suppose that $\dim H>2k>2$.
Let $f$ be an orthogonality preserving transformation of  ${\mathcal G}_k(H)$.
Then the following two conditions are equivalent:
\begin{enumerate}
\item[(A)] $f$ is ortho-adjacency preserving, i.e. 
for any ortho-adjacent $k$-dimensional subspaces  $X,Y\subset H$
the images $f(X), f(Y)$ are ortho-adjacent;
\item[(B)] for any adjacent $k$-dimensional subspaces  $X,Y\subset H$
the images $f(X), f(Y)$ are adjacent or $f(X)=f(Y)$. 
\end{enumerate}
If one of these conditions holds, then $f$ is induced by 
a linear or conjugate-linear isometry of $H$ to itself.  
\end{theorem}

The proof of Theorem \ref{theorem-wig-gr} is based on some properties of
the Grassmann graph $\Gamma_k(H)$ whose vertices are $k$-dimensional subspaces of $H$
and two vertices are connected by an edge if they are adjacent subspaces. 

A {\it clique} of a graph is a subset in the vertex set, where any two distinct vertices are adjacent.  
In the Grassmann graph $\Gamma_k(H)$ there are the following two types of maximal cliques:
\begin{enumerate}
\item[$\bullet$] the {\it star} ${\mathcal S}(X)$, where $X\in {\mathcal G}_{k-1}(H)$,
consists of all $k$-dimensional subspaces containing $X$;
\item[$\bullet$] the {\it top} ${\mathcal G}_{k}(Y)$, where $Y\in {\mathcal G}_{k+1}(H)$,
consists of all $k$-dimensional subspaces of $Y$.
\end{enumerate}
See \cite[Proposition 2.14]{Pankov-book}.

A subset of ${\mathcal G}_k(H)$ formed by mutually compatible subspaces
is said to be {\it compatible}. 

\begin{lemma}[Lemma 4.30 in \cite{Pankov-book}]\label{lemma-comp-star-top}
Every maximal compatible subset of a top contains precisely $k +1$ elements. 
Every maximal compatible subset of a star contains precisely $n-k + 1$ elements if 
$\dim H=n$ is finite, and this set is infinite if $H$ is infinite-dimensional.
\end{lemma}

A distance  $d(v,w)$ between to vertices $v$ and $w$ in a connected graph is defined as 
the smallest number $m$ such that there is a path consisting of $m$ edges and connecting $v$ with $w$;
every path connecting $v$ with $w$ and consisting of $d(v,w)$ edges is called {\it geodesic}.
The Grassmann graph $\Gamma_k(H)$ is connected and the distance between 
$k$-dimensional subspaces $X,Y\subset H$ in this graph is equal to 
$$k-\dim(X\cap Y)=\dim (X+Y)-k;$$
in particular, the distance between orthogonal $k$-dimensional subspaces is $k$.

\begin{lemma}[Lemma 4.31 in \cite{Pankov-book}]\label{lemma-geodesic}
Every geodesic in $\Gamma_k(H)$ connecting orthogonal subspaces consists of mutually compatible subspaces. 
Any two compatible $k$-dimensional subspaces $X,Y\subset H$ are contained in a certain geodesic of 
$\Gamma_k(H)$ connecting $X$ with a subspace orthogonal to $X$.
\end{lemma}

We start to prove Theorem \ref{theorem-wig-gr}.
Let $f$ be an orthogonality preserving transformation of ${\mathcal G}_k(H)$.

\begin{lemma}\label{lemma-2.1}
The conditions {\rm(A)} and {\rm(B)} from Theorem \ref{theorem-wig-gr} are equivalent. 
\end{lemma}

\begin{proof}
${\rm(A)} \Longrightarrow {\rm(B)}$.
Let $X$ and $Y$ be adjacent $k$-dimensional subspaces of $H$.
Then 
$$\dim(X +Y)=k+1\;\mbox{ and }\;\dim(X+Y)^{\perp}\ge 2$$
(since $\dim H >2k > 2$). 
We take orthogonal $1$-dimensional subspaces $P,Q\subset (X+Y)^\perp$
and consider the $k$-dimensional subspaces 
$$X' =(X\cap Y)+P\;\mbox{ and }\;Y' =(X\cap Y)+Q.$$
The subspaces $X,X',Y'$  are mutually ortho-adjacent and the same holds for
the subspaces $Y,X',Y'$.
Let ${\mathcal X}$ be a maximal compatible subset of the star ${\mathcal S}(X\cap Y)$ containing  $X,X',Y'$. 
Since $f$ is ortho-adjacency preserving,  
$f({\mathcal X})$ is a compatible subset in a star or a top. 
The assumption $\dim H>2k$ together with 
Lemma \ref{lemma-comp-star-top} implies that $f({\mathcal X})$ cannot be contained in a top, 
i.e. it is a subset of a star. 
Therefore, $f(X)$ contains the $(k-1)$-dimensional subspace $f(X')\cap f(Y')$. 
The same arguments show that $f(X')\cap f(Y')$ also is contained in $f(Y)$.
This means that $f(X),f(Y)$ are adjacent or $f(X)=f(Y)$.

${\rm(B)} \Longrightarrow {\rm(A)}$.
Let $X$ and $Y$ be $k$-dimensional subspaces of $H$.
The condition (2) guarantees that $f$ transfers every geodesic of $\Gamma_k(H)$ connecting $X$ and $Y$
to a path of $\Gamma_k(H)$ connecting $f(X)$ and $f(Y)$;
in particular, the distance between $f(X)$ and $f(Y)$ in  $\Gamma_k(H)$ is not greater than 
the distance between $X$ and $Y$. 

Suppose that $X,Y$ are orthogonal. Then $f(X),f(Y)$ are orthogonal 
and both these distances are equal to $k$. 
In this case, $f$ transfers every geodesic of $\Gamma_k(H)$ connecting $X$ and $Y$
to a geodesic connecting $f(X)$ and $f(Y)$.
Lemma \ref{lemma-geodesic} implies {\rm(A)}.
\end{proof}

From this moment, we assume that one of the conditions {\rm(A)} or {\rm(B)} holds. 
Then the other also is satisfied. 

The condition (B) guarantees that $f$ sends 
maximal cliques of $\Gamma_{k}(H)$ (stars and tops) to cliques.
Using the condition (A) and Lemma \ref{lemma-comp-star-top}, we show that
for every star ${\mathcal S}\subset {\mathcal G}_k(H)$ there is a star ${\mathcal S}'\subset {\mathcal G}_k(H)$
such that $f({\mathcal S})\subset {\mathcal S}'$. 
Since the intersection of two distinct stars contains at most one element and 
every star of $\Gamma_{k}(H)$ contains ortho-adjacent elements whose images are distinct,
such a star ${\mathcal S}'\subset {\mathcal G}_k(H)$ is unique. 
Therefore, for every $(k-1)$-dimensional subspace $X\subset H$
there is a unique $(k-1)$-dimensional subspace $X'\subset H$ such 
that $f({\mathcal S}(X))$ is contained in ${\mathcal S}(X')$.
In other words, $f$ induces a transformation
$f_{k-1}$ of ${\mathcal G}_{k-1}(H)$ satisfying
$$f({\mathcal S}(X))\subset {\mathcal S}(f_{k-1}(X))$$
for every $X\in {\mathcal G}_{k-1}(H)$. 
The latter inclusion implies that 
$$f_{k-1}({\mathcal G}_{k-1}(Y))\subset {\mathcal G}_{k-1}(f(Y))$$
for every $Y\in {\mathcal G}_{k}(H)$.

\begin{lemma}\label{lemma-2.2}
The following assertions are fulfilled:
\begin{enumerate}
\item[(1)] $f_{k-1}$ is orthogonality preserving;
\item[(2)] for any adjacent $X,Y\in {\mathcal G}_{k-1}(H)$
the images $f_{k-1}(X), f_{k-1}(Y)$ are adjacent or coincident;
\item[(3)] $f_{k-1}$ is ortho-adjacency preserving.
\end{enumerate}
\end{lemma}

\begin{proof}
(1). Suppose that $X,Y$ are orthogonal $(k-1)$-dimensional subspaces of $H$.
There are orthogonal $k$-dimensional subspaces $X',Y'\subset H$
such that $X\subset X'$ and $Y\subset Y'$. 
Then 
$$f_{k-1}(X)\subset f(X')\;\mbox{ and }\;f_{k-1}(Y)\subset f(Y').$$
The subspaces $f(X'),f(Y')$ are orthogonal and the same holds for $f_{k-1}(X),f_{k-1}(Y)$.

(2).  
If $X,Y$ are adjacent $(k-1)$-dimensional subspaces of $H$,
then ${\mathcal S}(X)\cap{\mathcal S}(Y)$ consists of one element.
Since 
$$f({\mathcal S}(X)\cap{\mathcal S}(Y))\subset {\mathcal S}(f_{k-1}(X))\cap {\mathcal S}(f_{k-1}(Y)),$$
the intersection of  ${\mathcal S}(f_{k-1}(X))$ and ${\mathcal S}(f_{k-1}(Y))$ is non-empty.
The latter is possible only in the case when $f_{k-1}(X), f_{k-1}(Y)$ are adjacent or $f_{k-1}(X)=f_{k-1}(Y)$.

The assertion (3) follows immediately from (1) and (2).
\end{proof}

Recursively, we constructs a sequence of transformations
$f_{k-i}$ of ${\mathcal G}_{k-i}(H)$ with $i=0,1,\dots,k-1$
such that $f_{k}=f$ and 
$$f_{k-i+1}({\mathcal S}(X))\subset {\mathcal S}(f_{k-i}(X))$$
for every $X\in {\mathcal G}_{k-i}(H)$ and 
$$f_{k-i}({\mathcal G}_{k-i}(Y))\subset {\mathcal G}_{k-i}(f_{k-i+1}(Y))$$
for every $Y\in {\mathcal G}_{k-i+1}(H)$ if $i\ge 1$.
In particular, $f_{1}$ is a lineation of ${\mathcal P}(H)$ to itself.
The direct analogue of Lemma \ref{lemma-2.2} holds for every $f_{k-i}$ with $i<k-1$
and $f_1$ is orthogonality preserving. 
By Theorem \ref{theorem-lin}, $f_1$ is induced by a linear and conjugate-linear isometry $L:H\to H$.
Since for every $X\in {\mathcal G}_{k}(H)$ we have
$$f_{1}({\mathcal G}_{1}(X))\subset {\mathcal G}_{1}(f(X)),$$
$f$ also is induced by $L$.
This completes the proof of  Theorem \ref{theorem-wig-gr}.

\subsection*{Acknowledgment}
The authors acknowledge the support by the Austrian Science Fund (FWF): project I 4579-N and 
the Czech Science Foundation (GA\v CR): project 20-09869L.

\end{document}